\UseRawInputEncoding
\documentclass[11pt,hyp]{article}
\usepackage{natbib,hyperref}
\usepackage{doi}
\usepackage[leqno]{amsmath}
\usepackage{amssymb,amsthm,graphicx,enumitem,upref}

\usepackage[titletoc, toc, title]{appendix}
\usepackage{breakcites}
\numberwithin{equation}{section}
\theoremstyle{plain}
\newtheorem{thm}{Theorem}[section]

\theoremstyle{definition} 
\newtheorem{exam}[thm]{Example}
\newtheorem{dfn}[thm]{Definition}
\newtheorem{prop}[thm]{Proposition}

\begin{document}

\title{An Evolutionary Theory for the Variability Hypothesis}

\author{Theodore P. Hill}

\date{\vspace{-5ex}}  

\maketitle

\begin{abstract}
An elementary biostatistical theory based on a ``selectivity-variability" principle is proposed to address a question raised by Charles Darwin, namely,  how one sex of a sexually dimorphic species might tend to evolve with greater variability than the other sex. Briefly, the theory says that if one sex is relatively selective then from one generation to the next, more variable subpopulations of the opposite sex will generally tend to prevail over those with lesser variability.  Moreover, the perhaps less intuitive converse also holds -- if a sex is relatively non-selective, then less variable subpopulations of the opposite sex will  prevail over those with greater variability. This theory requires certain regularity conditions on the distributions, but makes no assumptions about differences in means between the sexes, nor does it presume that one sex is selective and the other non-selective. Two mathematical models of the selectivity-variability principle are presented: a discrete-time one-step probabilistic model of short-term behavior with an example using normally distributed perceived fitness values; and a continuous-time deterministic model for the long-term asymptotic behavior of the expected sizes of the subpopulations with an example using exponentially distributed fitness levels.
\end{abstract}

\section{Introduction}
In his research on evolution in the 19th century Charles Darwin observed differences in variability between the sexes, reporting that 
\begin{quote}
throughout the animal kingdom, when the sexes differ in external appearance, it is, with rare exceptions, the male which has been the more modified; for, generally, the female retains a closer resemblance to the young of her own species, and to other adult members of the same group \cite[pp.~221]{vob79}. 
\end{quote}
Since then evidence of greater male variability, although by no means universal in either traits or species, has been reported in a wide variety of animal species from wasps and adders to salmon and orangutans (cf. \cite{vob105}). Specifically citing Darwin's research on animals  \cite[pp.~221--27]{vob79} and Ellis's research on humans \cite[pp.~358--372]{vob11} psychologist Stephanie Shields~wrote 
\begin{quote}
By the 1890's several studies had been conducted to demonstrate that variability was indeed more characteristic of males \ldots The biological evidence overwhelmingly favored males as the more variable sex \cite[pp.~772-73]{vob67}. 
\end{quote}

The past quarter century has produced much new research on the greater male variability hypothesis in different contexts, most of which refer to humans, and the majority of which support Darwin's observation (e.g., see Appendix A). After citing specific evidence of greater male variability, Darwin had also raised the question of {\em why} this might occur, writing 
\begin{quote}
The cause of the greater general variability in the male sex, than in the female is unknown \cite[p.~224]{vob79}. 
\end{quote}
This question has persisted into the 21st century as noted for example by Hyde et al.: ``There is evidence of slightly greater male variability in scores, although the causes remain unexplained'' \cite[p.~495]{vob28}, and Halpern et al.: ``the reasons why males are often more variable remain elusive'' \cite[p.~1]{vob25}.
As  statistician Howard Wainer phrased it, 
\begin{quote}
Why was our genetic structure built to yield greater variation among males than females? And not just among humans, but virtually all mammals \cite[p.~255]{vob60}.
\end{quote}

The objective of this paper is to propose an elementary mathematical principle to help explain how  one sex in a given species could naturally evolve toward greater or lesser variability depending on the preferences of the opposite sex. Together with two additional standard biological tenets, this principle might help provide an answer to Darwin's question.
 

 \section{A Theory for Differences in Variability Between Sexes}
In very general terms, the first principle of the theory introduced here is this:
  \vspace{1em}
  
\noindent
\textbf{\textsc{Selectivity-Variability Principle.}} 
{\em 
In a species with two sexes A and B, both of which are needed for reproduction, suppose that sex A is relatively {\bf selective}, i.e., will mate only with a top tier (less than half) of B candidates. Then from one generation to the next, among subpopulations of B with comparable average attributes, those with {\bf greater variability} will tend to prevail over those with lesser variability. Conversely, if A is relatively {\bf non-selective}, accepting all but a bottom fraction (more than half) of the opposite sex, then subpopulations of B with {\bf lesser variability} will tend to prevail over those with comparable means and greater variability.
}
\\

Note that this principle does not make any assumptions about inherent differences in means or other attributes between the sexes.  For instance, it does not presume that one sex is selective and the other non-selective, or even that one sex is more selective than the other, unlike Bateman's principle \cite{novob11}, for example, or other related theories such as ``the sex that experiences more intense...vetting by the other sex will tend to show greater within-sex variation on many traits" \cite[p.~176]{vob3}.  If both sexes of a species happen to be selective, the selectivity-variability principle here predicts that the best evolutionary strategy for each is to tend toward greater variability.

It is also important to note that this principle alone says nothing about {\em a priori} or {\em a posteriori} comparisons of the variabilities across the sexes. For example, if sex A is not selective while sex B is selective, this principle says that subpopulations of A with greater variability will prevail over subpopulations of A with lesser variability, and that subpopulations of B with lesser variability will prevail over subpopulations of sex B with greater variability. It says nothing about comparing the resulting variability of sex A with the variability of sex B. If all the subpopulations of B were initially more variable than all the subpopulations of A, for instance, then the next generation of sex B will still exhibit greater variability than the next generation of sex A whether either sex is selective or non-selective. Only under additional hypotheses, such as ``both sexes began with comparable mid-range variability", as will be done in an illustrative example below to address Darwin's question, can this theory be useful to draw any conclusions about comparisons of variability between the two sexes.

 \vspace{1em}
 
 In order to make this selectivity-variability theory more precise, of course, it is necessary to define formally what is meant by selectivity and variability in this context, and that will be done in the next section. 
 First, the following simple informal hypothetical example may help convey the intuition behind this principle.
 
 \vspace{1em}
 
  \begin{figure}[!ht] 
 
  \center\includegraphics[width=0.75\textwidth, height=0.75\textheight]{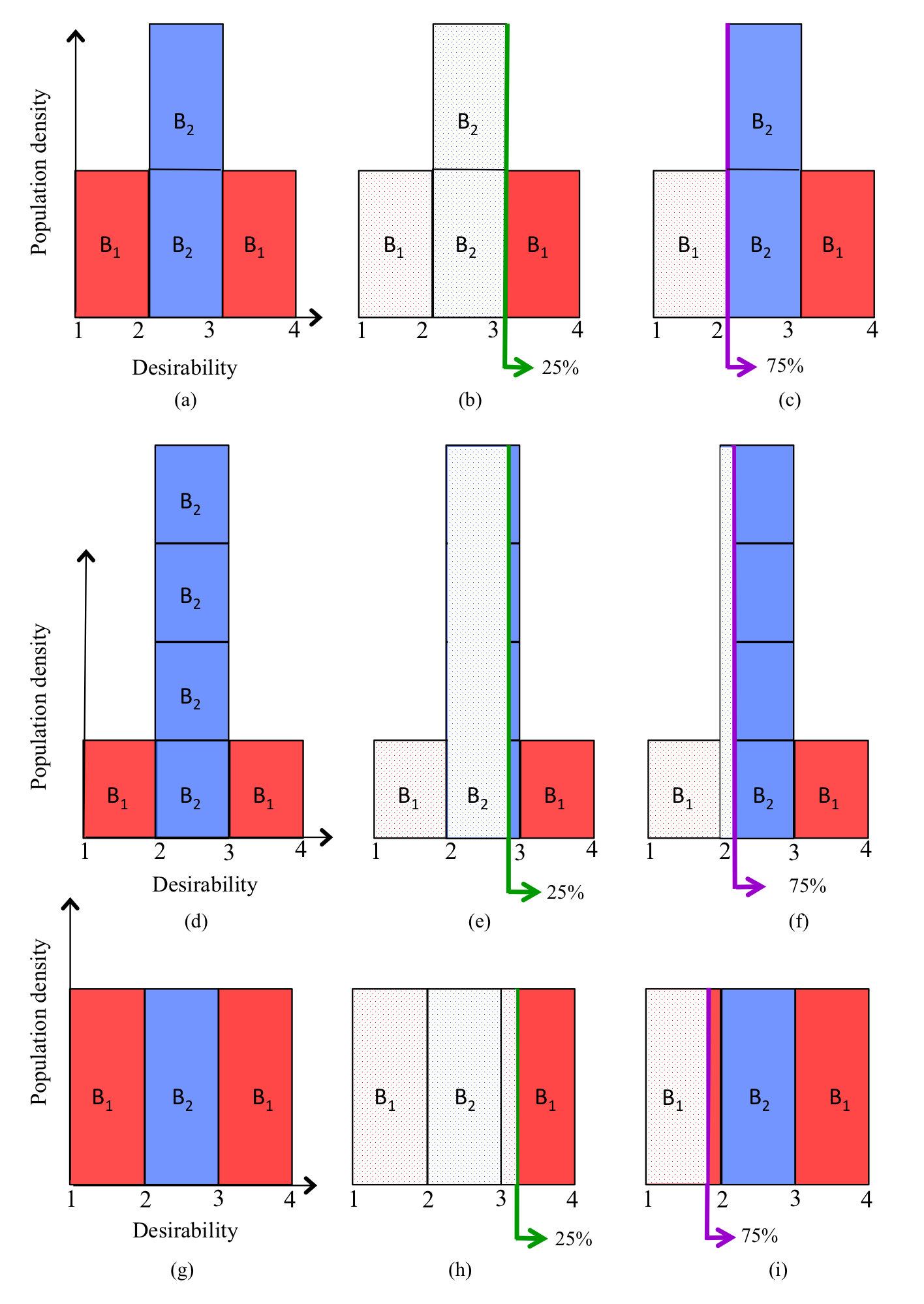}
  \caption{The three cases of Example~\ref{exam1}:  The red blocks represent the distribution of the more variable subpopulation $B_1$ of sex $B$, and the blue represent the less variable subpopulation $B_2$.}
  \label{newFig1}
\end{figure} 

  \begin{exam} \label{exam1} 
  Sex $B$ consists of two subpopulations $B_1$ and $B_2$. Sex $A$ considers half of the individuals in $B_1$ very desirable and the other half not very desirable, and it considers all of the individuals in subpopulation $B_2$ of mid-range desirability. Thus  $B_1$ is more variable in desirability to sex $A$ than $B_2$, with $B_1$ and $B_2$ having comparable average desirability. (See Figure~\ref{newFig1}, where larger numbers reflect greater desirability.)
  
{\em Special Case 1.} $B_1$ and $B_2$ are of equal size. Then sex $A$ considers one quarter of sex $B$ (the lower half of $B_1$) of relatively low desirability, half of $B$ (all of $B_2$) of medium desirability, and one quarter of $B$ (the upper half of $B_1$) of above-average desirablity (see Figure~\ref{newFig1}a). If sex  $A$ is relatively selective and will mate only with the top most desirable quarter of sex
$B$, then all of the next generation will be offspring of the more variable subpopulation $B_1$ 
(Figure~\ref{newFig1}b).
On the other hand, if sex $A$ is relatively non-selective and will mate with any but the lower quarter of $B$, then all of the less variable $B_2$ will mate, but only half of the more variable $B_1$ will mate (Figure~\ref{newFig1}c).

Similar conclusions follow if the initial subpopulations are not of equal size. 

{\em Special Case 2.} One third of sex $B$ is the more variable $B_1$ and two thirds is the less variable $B_2$ (Figure~\ref{newFig1}d). If sex  $A$ only mates with the most desirable quarter of $B$, a short calculation shows that two thirds of the next generation will be offspring of $B_1$ and one third will be offspring of $B_2$, so based on the initial distribution, the more variable subpopulation will be overrepresented (Figure~\ref{newFig1}e). If sex $A$ will mate with any but the least desirable quarter of $B$, then only two ninths of the next generation will be offspring of $B_1$ and seven ninths will be offspring of  $B_2$, so the less variable subpopulation of sex  $B$ will be overrepresented (Figure~\ref{newFig1}f).

{\em Special Case 3.} Two thirds of sex $B$ is the more variable $B_1$ and only one third is the less variable $B_2$ (Figure~\ref{newFig1}g). If sex $A$ only mates with the most desirable quarter of $B$, then all of the next generation will be offspring of $B_1$ (Figure~\ref{newFig1}h), and if sex $A$ will mate with any but the least desirable quarter of $B$, then only five ninths of the next generation will be offspring of $B_1$ and the rest will be offspring of $B_2$, so again the less variable subpopulation of sex $B$ will be overrepresented (Figure~\ref{newFig1}i).
\end{exam}

\noindent
 Note the asymmetry in the mating probabilities in this example; some intuition behind why this occurs may perhaps be gained from the observation that the most desirable individuals in the more variable population will always be able to mate, whether the opposite sex is selective or non-selective. 

\section{Desirability and Selectivity}

In order to begin to try to interpret these ideas analytically, it is of course necessary to identify concrete definitions of ``desirable", ``selective" and ``more variable". There are clearly many different candidates to capture the essence of each of these terms; the following assumptions and definitions are simply intended as a starting point to facilitate proposal of several models and analysis of the above selectivity-variability principle. 

To begin with, the informal notion of desirability introduced in Example \ref{exam1} above will be extended as follows.

 \vspace{1em}
 
 \textbf{\textsc{Desirability Assumption.}} 
{\em 
Each individual (or phenotype) in each sex is assigned a real number which reflects its desirability to the opposite sex, with higher values indicating greater desirability.}
\\

 As a concrete example, one interpretation of the desirability value of an individual might be the opposite sex's perception or estimation of its Darwinian fitness (e.g., \cite{novob1}). The actual magnitudes of these desirability values are not assumed to have intrinsic significance in general, but are used only to make comparisons between individuals. Here and throughout, it will also be assumed that the same desirability value is assigned to each individual by every member of the opposite sex. In real life scenarios, of course, the desirability of an individual varies from one member of the opposite sex to another, and is not quantifiable in a single one-parameter value. 
 
Next, the informal notion of selectivity introduced in Example \ref{exam1} above will be formalized.
\vspace{1em}

\textbf{\textsc{Selectivity Assumption.}} 
{\em 
For each sex in a given sexually dimorphic species there is an upper proportion $p \in (0,1)$ of the opposite sex that is acceptable for mating. }\\

For example, if $p_A$ is that proportion for sex $A$, then members of sex $A$ will mate with individual $b$ of the opposite sex $B$ if and only if $b$ is in the most desirable $p_A$ fraction of individuals in $B$. If $p_A<0.5$, then sex  $A$ is said to be (relatively) {\em selective}, and if $p_A>0.5$, then  $A$ is said to be {\em non-selective}. For instance, if $p_A=0.25$, then
sex $A$ is selective, since it will mate only with the most desirable quarter of sex $B$, and if $p_A=0.75$, then sex $A$ is non-selective, 
since it will mate with any but the least desirable quarter of sex $B$.

\vspace{1em}

\noindent
{\em N.B.} Of course these assumptions about desirability values and selectivity are clearly not satisfied in most real life scenarios, and are simply intended here as a starting point for discussion of the general ideas.  For example, the acceptability fractions $p_A$ may reflect not only desirability, but also availability or proximity. In this simple model it is therefore assumed that the populations are large and mobile so there are always available potential mates of the opposite sex above the threshold desirability cutoff. Similarly, for simplicity it will be assumed throughout that the offspring of any coupling consist of equal numbers of each sex.

\section{Variability}

The desirability of individuals in one sex by the opposite sex varies from individual to individual, and its normalized distribution is a probability distribution. Thus to address the notion of differences between two subpopulations of the same sex in the variability of their desirability to the opposite sex, the notion of one probability distribution being more (or less) variable than another must be specified.  As will be seen in the next example, for instance, if by ``more variable" is meant ``larger standard deviation" (or statistical variance), then the above selectivity-variability principle is not true without additional assumptions on the underlying distributions.

\begin{exam} \label{examB} 
Sex $B$ consists of two subpopulations $B_1$ and $B_2$, with six individuals each: $B_1$ has one individual of desirability value $1$ (to sex $A$), one of desirability $5$, and four individuals of desirability $3$; $B_2$ has three individuals of desirability value $2$ and three of desirability $4$. Thus both $B_1$ and $B_2$ have mean desirability  $3$, the variance of $B_1$ is $4/3$ and the variance of $B_2$ is $1$.  

If sex $A$ is {\em selective} with $p_A = 0.25$, then two of the three individuals that sex $A$ selects from sex $B$ will be from $B_2$, the subpopulation with {\em smaller variance}. Conversely, if sex $A$ is {\em non-selective} with $p_A = 0.75$, then five of the nine individuals that sex $A$ selects from sex $B$ will be from $B_1$, the subpopulation with {\em larger variance}.  Thus for these distributions and a standard deviation definition of variability, both directions of the above selectivity-variability principle fail.
\end{exam}

There are many other possibilities for definitions of variability, such as comparisons of ranges or Gini mean differences, but those can be very misleading in this setting since a single outlier can dramatically alter the values of such statistics. On the other hand, basic comparisons of the tails of distributions leads to a natural notion of greater or lesser variability.

To that end, for a real Borel  probability measure $P$ let $S_P$ denote the complementary cumulative distribution function of $P$,  i.e.,  $S_P:\mathbb{R}\to [0,1]$  is defined by $S_P(x)=P(x,\infty)$ for all $x\in \mathbb{R}$.  That is, $S_P(x)$ is simply the proportion of a population with distribution $P$ that is above the threshold $x$; see Figure~\ref{newFig2} for three examples. For brevity, the term {\em survival function} will be used here; in this context  $S_P(x)$ may be thought of as the proportion of a given sex with desirability (by the opposite sex) distribution function $P$ that ``survives" the cut when the opposite sex only accepts individuals  with desirability value $x$ or larger.

\begin{dfn} \label{def1}
$P_1 \succ P_2$ if, for all $x$ with $0 < S_{P_1}(x) < 1$,
$$
S_{P_1}(x)> S_{P_2}(x) \mbox{ for all } x>m \mbox{ and }S_{P_1}(x)<S_{P_2}(x)\mbox{ for all } x<m.
$$  
\end{dfn}

In other words, $P_1$  {\em is more variable than} $P_2$ if the proportions of $P_1$ both above every upper (larger than median) threshold and the proportions below every lower threshold level are greater than those for $P_2$. That is, both upper and lower tails of the $P_1$ distribution are heavier than those of the $P_2$ distribution, for all thresholds.

\begin{figure}[!ht] 
  \center\includegraphics[width=1.0\textwidth]{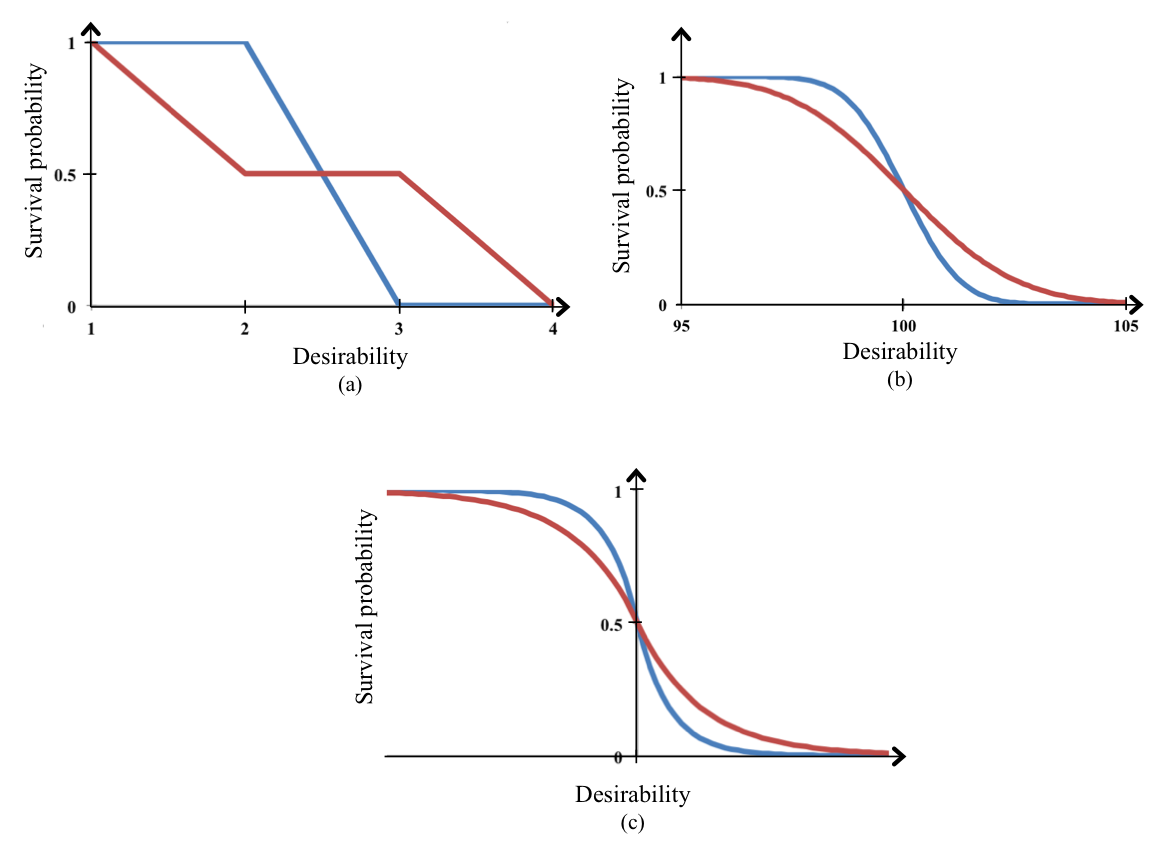}
  \caption{The survival functions and comparative variability of three pairs of distributions:  (a) the uniform distributions in Example~\ref{exam1} above;  (b) the normal distributions in Example~\ref{exam2} below; and (c) the Laplace distributions in Example~\ref{exam3}.  In each case, the red curves denote the more variable distribution.}
 \label{newFig2}
\end{figure}

In Example~\ref{exam1}, where the selectivity-variability principle was illustrated informally, the distribution of subpopulation $B_1$ is more variable than the distribution of subpopulation $B_2$ both in the sense of standard deviation and in the sense of Definition~\ref{def1} (see Figure~\ref{newFig2}a), and it is this definition that will be seen below to lead to settings where the principle is valid. As was seen in Example~\ref{examB}, the selectivity-variability principle may fail for arbitrary distributions if variability is defined in terms of standard deviation, but the next proposition identifies several common and important classes of distributions where greater standard deviation or scale factor (when standard deviation is infinite) coincide with the notion of greater variability in Definition~\ref{def1}, and thus these distributions are applicable to the models below. The conclusions are perhaps well-known, but as no reference is known to the author, a short proof is included. 

\begin{prop} \label{newPropA} 
Let $P_1$ and $P_2$ be (real Borel) probability measures with identical medians.
\begin{enumerate}[label=(\roman*)]
\item If both $P_1$ and $P_2$ are uniform, symmetric triangular, Laplace, or Gaussian, then
\begin{equation*}
P_1 \succ P_2 \quad \mbox{if and only if} \quad variance(P_1) > variance(P_2).
\end{equation*}

\item If  both $P_1$ and $P_2$ are Cauchy, then 
\begin{equation*}
P_1 \succ P_2 \quad \mbox{if and only if} \quad \mbox{ {\em scale factor} of } P_1 > \mbox{ {\em scale factor} of } P_2.
\end{equation*}
\end{enumerate}
\end{prop}

\begin{proof}
Since the cumulative distribution functions for uniform, symmetric triangular, Laplace, and Cauchy distributions are known in closed form, the conclusions regarding those distributions follow from Definition \ref{def1} and routine calculations by comparisons of the respective piecewise linear, quadratic, exponential, and arcsin distribution functions.

To see the conclusion for Gaussian distributions, for which the distribution functions are not known in closed form, suppose
$X_1\sim N(\mu,\sigma^2_1)$ and $X_2\sim N(\mu,\sigma^2_2)$, where $N(\mu,\sigma^2)$ denotes a normal distribution with mean $\mu$ and standard deviation $\sigma$. Without loss of generality suppose that $\sigma^2_1>\sigma^2_2$.   Then for all $c>\mu$,
\begin{equation*}\label{neweq1} 
\begin{split} 
P(X_1>c)
&= P(\sigma_2(X_1-\mu)>\sigma_2(c-\mu))\\
&= P(\sigma_1(X_2-\mu)>\sigma_2(c-\mu))\\
&>P(\sigma_1(X_2-\mu)>
\sigma_1(c-\mu))\\
&= P(X_2>c),
\end{split}
\end{equation*}
where the second equality follows since, by the rescaling and translation properties of normal distributions, 
\begin{equation*}\label{neweq3}  
\sigma_2(X_1-\mu) \mbox{ and }\sigma_1(X_2-\mu) \mbox{ are both }
N(0,\sigma^2_1\sigma^2_2).
\end{equation*}
The case $c < \mu$ follows similarly, and since the mean of every normal distribution is the same as the median, this completes the proof.
\end{proof}

It should also be noted that distributions sufficiently close to the distributions in Proposition~\ref{newPropA}  will also obey the same variability conclusions. (E.g., no real-life data is ever exactly Gaussian, but in many applications Gaussian distributions are good approximations and very useful in practice.) Note that the above definition of greater variability does not require finite standard deviations or symmetry of the distributions, although the examples provided below have both. Some assumption on two distributions (of the same sex) having comparable average attributes is clearly necessary to be able to draw any useful conclusions in this selectivity context; the assumption of identical medians used here is one natural candidate. Similar conclusions may be drawn about weak-inequality versions of this definition and about one-sided variability, and these are left to the interested reader. For example, if both the median and upper tails of one distribution are larger than those of another, then that first distribution will also prevail if the opposite sex is selective.

\vspace{1em}

Using the above definitions of variability and selectivity, the main objective of this paper is to present two mathematical models for the selectivity-variability principle above. 

\section{A Discrete-Time Probabilistic Model}
Suppose  that sex  $B$ of a given hypothetical species consists of two distinct subpopulations $B_1$ and $B_2$, of which a proportion $\beta \in (0,1)$ is of type $B_1$ (and $1-\beta$ is of type $B_2$). Let $P_1$ and $P_2$ denote the desirability distributions of $B_1$ and $B_2$, respectively, and assume that $P_1$ is more variable than $P_2$, i.e., $P_1 \succ P_2$.  It will now be shown that, for all $\beta$,  if sex $A$ is selective, then subpopulation $B_1$ will be overrepresented in the subsequent generation, and if sex $A$ is non-selective, then subpopulation $B_2$ will be overrepresented in the subsequent generation.  These are direct analogs and extensions of the informal observations in Example~\ref{exam1} above. 

To that end, note that if $\beta$ is the proportion of sex $B$ that is from subpopulation $B_1$, then letting $S_1$ and $S_2$ denote the survival functions of $B_1$ and $B_2$, respectively, the number
$$
\frac{\beta S_1(c)} { \beta S_1(c)+(1-\beta)S_2(c)}
$$	
represents the proportion of sex $B$ that is from subpopulation $B_1$ when $A$ accepts only individuals in $B$ with desirability value above cutoff level $c$. This motivates the following definition.
	
\begin{dfn} \label{def2}
If sex $B$ consists of two subpopulations $B_1$ and $B_2$  and if $\beta \in (0,1)$ is the proportion of sex $B$ that is $B_1$, then {\em subpopulation $B_1$ will be overrepresented in the subsequent generation} if and only if
$$
\frac{\beta S_1(c^*)} { \beta S_1(c^*)+(1-\beta)S_2(c^*)} > \beta,
$$
where $S_1$ and $S_2$ are the survival functions of the desirability distributions of $B_1$ and $B_2$, respectively, and $c^*$ is the desirability cutoff of sex $A$ for mating with individuals in sex $B$, i.e., 
$$
 \beta S_1(c^*)+(1-\beta)S_2(c^*) = p_A.
 $$
\end{dfn}

Note that this definition does not assume that the offspring of $B_1$ will have desirability distributions identical to that of $B_1$ but simply that a larger proportion of the subsequent generation will be offspring of $B_1$ than the proportion of $B_1$ in the original population. With this notion of overrepresentation, an elementary formalized version of the selectivity-variability principle is as follows.  

\begin{thm}\label{propA2} 
Let  sex  $B$ consist of two distinct subpopulations $B_1$ and $B_2$ with desirability distributions $P_1$ and $P_2$, respectively, with identical medians $m$ and with desirability survival functions $S_1$ and $S_2$ which are continuous and strictly decreasing. 
Suppose subpopulation $B_1$ is more variable than $B_2$, i.e., $P_1 \succ P_2$.  Then

\begin{enumerate}[label=(\roman*)]
\item If $p_A<0.5$ , i.e., if sex $A$ is selective, then the more variable subpopulation $B_1$ will be overrepresented in the subsequent generation.

\item If $p_A>0.5$, i.e., if sex $A$ is non-selective, then the less variable subpopulation $B_2$ will be overrepresented in the subsequent generation.

\end{enumerate}
\end{thm}

\begin{proof} 
Let $\beta \in (0,1)$ be the proportion of $B$ that is $B_1$, 
and let $S_1$ and $S_2$ denote the desirability survival functions for $B_1$ and $B_2$, respectively.  First, it will be shown that there exists a unique ``threshold" desirability cutoff $c^*\in\mathbb{R}$
such that
\begin{equation}\label{eq1} 
\begin{split}
& \beta S_1(c^*)+(1-\beta)S_2(c^*)=p_A \\
&\mbox{ and } \\
& c^*>m \mbox{ if } p_A<0.5 \mbox{ and }c^*<m \mbox{ if }p_A>0.5. 
\end{split}
\end{equation}

To see \eqref{eq1}, let
$g:\mathbb{R}\to (0,1)$ be given by 
$g(c)=\beta S_1(c)+(1-\beta)S_2(c)$. Then  
$g$ is continuous and strictly decreasing
with $g(-\infty)=1$, $g(m)=0.5$, $g (\infty)=0$, so $c^*$ satisfying \eqref{eq1} exists and is unique, and since 
$S_1(m)=S_2(m)=0.5$, $c^*>m$
 if $p_A<0.5$ and $c^*<m$ if $p_A>0.5$.	 

To see (i),  first note by \eqref{eq1} that $c^* > m$, so since $P_1 \succ P_2$, 
\begin{equation}\label{newesteq2}
S_1(c^*) > S_2(c^*).
\end{equation}
Thus, 
$$
\beta(1-\beta) S_1(c^*) > \beta(1-\beta) S_2 (c^*), 
$$
which implies
$$
\beta^2 S_1 (c^*) + \beta(1-\beta) S_1 (c^*) > \beta^2 S_1 (c^*) + \beta(1-\beta) S_2 (c^*)
$$
so
$$
\frac{\beta S_1(c^*)} { \beta S_1(c^*)+(1-\beta)S_2(c^*)} > \beta.
$$

By Definition~\ref{def2}, this completes the proof of (i);  the proof of (ii) follows similarly.
 
\end{proof}

Thus, in this discrete-time setting, if one sex remains non-selective from each generation to the next, for example, then in each successive generation less variable subpopulations of the opposite sex will tend to prevail over more variable subpopulations of comparable average desirability. Although  those successive generations and their desirability distributions are evolving over time, if less variable subpopulations in the opposite sex prevail over more variable subpopulations from each generation to the next, that suggests that over time the opposite sex will tend toward lesser variability. A key assumption here, of course, is the inheritability of variability itself.

That variability {\em per se} may be a heritable trait has recently been established in several different contexts (e.g., \cite{vob107}, \cite{vob117}, and \cite{vob95}). For instance, theories of inherited variability have been developed and applied by animal husbandry scientists and geneticists who are interested in breeding livestock, not only for high averages of desirable traits, but also for uniformity (i.e., low variability; see \cite{vob116}, \cite{vob118}, \cite{vob119}). The above probabilistic model only uses the premise that variability is inheritable; identification of the precise genetic, chromosomal, and epigenetic (including societal) mechanisms for how variability is inherited is beyond the scope of this paper, and is left to the interested reader (cf. \cite{vob115}, \cite{novob12}, \cite{vob48}, and \cite{vob154}).

The next example is an application of Theorem \ref{propA2} and Proposition \ref{newPropA} in a case where the desirability values have gaussian distributions. Note that the exact magnitudes of over- and underrepresentation, even in this standard setting, are not attainable in closed form, but must be estimated numerically.

\begin{exam}\label{exam2} 
Suppose that the desirability values (to sex $A$) of sex $B$ are normally distributed, i.e., if  $X_1$ and $X_2$ are the desirability values of two random individuals chosen from $B_1$ and $B_2$, 
respectively, then $X_1$ has distribution $N(\mu,\sigma^2_1)$ and $X_2$ has distribution $N(\mu,\sigma^2_2)$.
 (The assumption  of normality for the underlying distributions of desirability is not essential; this is merely an illustrative example, and chosen because of the ubiquity of the normal distribution in many population studies. Note the key assumption that the average values, i.e. the medians, are the same.) By Proposition~\ref{newPropA}, $N(\mu,\sigma^2_1)$ is more variable than
 $N(\mu,\sigma^2_2)$ if and only if $\sigma^2_1>\sigma^2_2$.
 
 In particular, suppose $X_1\sim N(100,4)$, 
$X_2\sim N(100,1)$, so $B_1$ is more variable than  $B_2$.  Suppose that $B_1$ and $B_2$ are of equal size, and again consider the two typical cases where sex $A$ is selective with $p_A=0.25$ and where sex $A$ is non-selective with $p_A=0.75$ (see Figure~\ref{fig3}). 

Suppose first that $p_A=0.25$. Using a special function calculator (since the survival functions of normal distributions are not known in closed form), it can be determined numerically that sex $A$'s threshold desirability value cutoff for sex $B$ is $c^*\cong 100.92$,  $S_1(c^*)\cong 0.323$, and $S_2(c^*)\cong 0.179$.  Thus a random individual from subpopulation $B_1$ has nearly twice the probability of mating than one from the less variable subpopulation $B_2$, as is illustrated in Figure~\ref{fig3} with the areas to the right of the green desirability cutoff. Hence $B_1$ will be overrepresented in the subsequent generation.

Next suppose that $p_A=0.75$. Then it can be determined that the threshold desirability value cutoff is $c^*\cong 99.08$, $S_1(c^*)\cong 0.677$, and $S_2(c^*)\cong 0.821$, i.e., a random individual from subpopulation $B_2$ is about one-fifth more likely to be able to mate than one from the more variable subpopulation $B_1$. This is illustrated in Figure~\ref{fig3} with the areas to the right of the purple cutoff. Here again, note the asymmetry in that the selective case is more extreme than the non-selective case, as was seen in Example~\ref{exam1}. 
\end{exam}

\begin{figure}[!ht] 
  \center\includegraphics[width=0.75\textwidth]{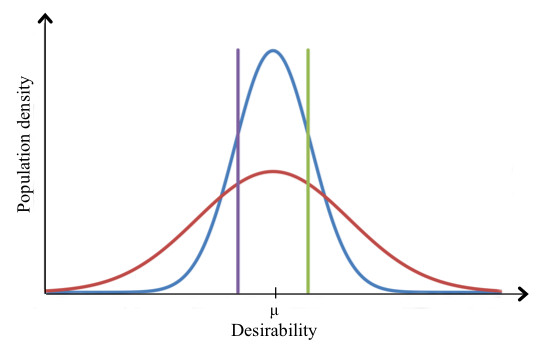}
  \caption{The red curve is the desirability distribution of the more variable normal subpopulation $B_1$ in Example~\ref{exam2} and the blue curve is the desirability distribution of the less variable subpopulation $B_2$. The vertical green line is the threshold cutoff for the opposite sex $A$ so that exactly 25\%  of the composite $B$ population has desirability value above (to the right of) that point. The vertical purple line is the value so that exactly 75\% of the  $B$ population has desirability value above that point.}
 \label{fig3}
\end{figure}

\section{A Continuous-Time Deterministic Model}

In this section, comparisons of the long-term asymptotic behavior of the sizes of competing subpopulations of the same sex and species are modeled using the general structure and logic of classical evolutionary game theory as applied to population dynamics (cf. \cite{novob14}, \cite{novob17}, \cite{novob16}).

Here, sex $B$ consists of two distinct subpopulations $B_1$ and $B_2$, growing in time, whose sizes at time $t$ are given by the continuous stochastic processes $X_1(t)$ and $X_2(t)$, respectively. Letting $x_1(t)=E[X_1(t)]$ and $x_2(t)=E[X_2(t)]$ denote the expected values of the subpopulation sizes at time $t$, the objective here will be to derive a coupled system of ODE's, directly analogous to the coupled systems of ODE's in classical evolutionary game theory, to model the growth rates of the expected values of the sizes of the two subpopulations (cf. \cite{novob2}).

In contrast to the discrete-time model above, here there is no clear delineation between generations, and it will be assumed that the pace of evolution is negligible compared to the pace of reproduction, so the two subpopulations remain distinct, with offspring distributed the same way as the parent subpopulation. In this setting, it will now be seen that if one subpopulation is more variable than the other, then the more variable subpopulation will eventually eclipse the less variable subpopulation if the opposite sex is selective, and the less variable subpopulation will eclipse the more variable one if the opposite sex is non-selective.  

Assume that the desirability distributions of $B_1$ and $B_2$ (to sex $A$) are given by probabilities $P_1$ and $P_2$, respectively, that do not change with the sizes of the subpopulations, i.e., the survival and desirability distribution functions do not change with $t$. For further ease of analysis, assume that the expected population sizes $x_1(t)$ and $x_2(t)$ are strictly increasing and differentiable and that the survival functions $S_1$ and $S_2$ for $P_1$ and $P_2$  are both continuous and strictly decreasing, with identical (unique) medians $m$. In other words, exactly half of each subpopulation $B_1$ and $ B_2$ has desirability value above $m$ to sex $A$ at all times $t>0$, and exactly half of each has desirability values below $m$. 

In this deterministic framework, the composite population of sex $B$ is growing at a rate that is proportional to the fraction $p_A$ of its members that is acceptable to the opposite sex $A$. That is, with the constant of proportionality taken to be 1, 
\begin{equation}\label{eq2}
\frac{d(x_1+x_2)}{dt} = p_A(x_1+x_2).
\end{equation}
Similarly, the expected values of the sizes of both subpopulations $B_1$ and $B_2$ are growing at rates proportional to the fractions $S_i(c^*)$ of each subpopulation that are acceptable to sex $A$ at that time, i.e., which satisfy the coupled system of ordinary differential equations 
\begin{equation}\label{eq3} 
\frac{dx_i}{dt} = x_iS_i(c^*),\quad i=1,2,
\end{equation}
where $c^*=c^*(x_1,x_2)$ is the cutoff value so that the expected proportion of sex $B$ that is above desirability level $c^*$ at time $t$ is exactly $p_A$. Since $P_1$ and $P_2$ and hence $S_1$ and $S_2$ are assumed to be constant in time, $c^*$ satisfies

\begin{equation}\label{eq9A} 
\frac{x_1S_1(c^*)+x_2S_2(c^*)}{x_1+x_2}=p_A.
\end{equation}

\vspace{1em}

The coupled system of ODE's \eqref{eq3} is very closely related to the classical replicator equation (cf.\cite{novob13}, \cite{novob15}), which also captures the essence of selection via acceptability for mating but through rates proportional to deviation from the mean desirability (or fitness), rather than through rates proportional to fractions above selectivity cutoffs.
Analogous to the discrete probabilistic model above, solutions of \eqref{eq3} are not generally available in closed form, and must be approximated numerically, as will be seen in Example~\ref{exam3} below. 

Qualitative comparisons of the rates of growth of competing subpopulations satisfying (4) are possible, however, and the next theorem shows that the selectivity-variability principle above is also valid in this setting, in the following sense. If $P_1$ is more variable than $P_2$,  and if  $p_A<0.5$, i.e., if sex $A$ is selective, then the relative instantaneous rate of growth of $B_1$ exceeds that of  $B_2$, and the proportion of sex $B$ that is from $B_1$ approaches 1 in the limit as time goes to infinity. Conversely, if $p_A>0.5$, i.e., if sex $A$ is non-selective, then the relative instantaneous rate of growth of $B_2$ exceeds that of $B_1$, and the less variable subpopulation $B_2$ prevails in the limit. This same conclusion can be extended to more general settings, such as time-dependent acceptability fractions $p_A(t)$, and these generalizations are left to the interested reader.

\vspace{1em}
Recall that $P_1$ and $P_2$, respectively, are the desirability distributions (to sex $A$) of subpopulations $B_1$ and $B_2$ of sex $B$.

\begin{thm}\label{propB1} 
Suppose subpopulation $B_1$ is more variable than $B_2$, i.e., $P_1 \succ P_2$. 
 \begin{enumerate}[label=(\roman*)]
\item If 
$p_A<0.5$, i.e., if sex $A$ is selective, then the relative rate of growth of $B_1$ exceeds that of 
$B_2$, 
\begin{equation}\label{eq4}  
\frac1{x_1}\frac{dx_1}{dt} > \frac1{x_2}\frac{dx_2}{dt}.
\end{equation}
Moreover,  
$\frac{x_1}{x_1+x_2}\to 1$ as $t\to\infty$.
\item If $p_A>0.5$, i.e., if sex $A$ is non-selective, then the relative
rate of growth of $B_2$ exceeds that of $B_1$, 
\begin{equation}\label{eq5}
\frac1{x_2}\frac{dx_2}{dt} > \frac1{x_1}\frac{dx_1}{dt}.
\end{equation}
Moreover,
$\frac{x_1}{x_1+x_2}\to 0$ as $t\to\infty$.
\end{enumerate}
\end{thm}

\begin{proof}[Proof of $(i)$] Analogous to the argument for 
Theorem~\ref{propA2}, define  $g:\mathbb{R}\to (0,1)$ by
\begin{equation*}\label{eq6}
g(c)=\frac{x_1S_1(c)+x_2S_2(c)}{x_1+x_2},
\end{equation*}
where $S_1$ and $S_2$
are the desirability survival functions for $P_1$ and $P_2$, respectively.  Recall that $S_1$ and $S_2$ are both continuous and strictly decreasing with identical medians $m>0$, and fix $t > 0$. Since $g$ is continuous
and strictly decreasing with $g(-\infty)=1$, 
$g(m)=0.5$, and $g(\infty)=0$,  there exists a unique threshold
desirability cutoff	$c^*=c^*(t)$ satisfying 
\begin{equation*}\label{eq7}
\frac{x_1S_1(c^*)+x_2S_2(c^*)}{x_1+x_2}=p_A,
\end{equation*}
where, as before, $p_A$ is the most desirable fraction of sex $B$ that is acceptable to
sex $A$, and $c^*=c^*(t)$ is the threshold desirability cutoff for sex 
$A$ for the combined
populations of sex $B=B_1\cup B_2$ at time $t$.

Note that $S_1(m)=S_2(m)=0.5$, so since $p_A<0.5$, $c^*>m$. Since 
$P_1$ is more variable than $P_2$ this
implies that $S_1(c^*)>S_2(c^*)$. 
Since $S_1(c^*)$ and $S_2(c^*)$ are the proportions of 
$B_1$ and $B_2$, respectively,
that are above the threshold cutoff at time $t >0$,  \eqref{eq3} implies \eqref{eq4}.

To see that  $\frac{x_1}{x_1+x_2}\to 1$ as $t\to\infty$, 
note that since $P_1$ is more variable than $P_2$,  $m < S_{2}^{-1}(p_A)<S_{1}^{-1}(p_A)$ for
$p_A<0.5$. Clearly $c^*\in[S_{2}^{-1}(p_A),S_{1}^{-1}(p_A)]$ for all 
$t > 0$, so since $S_2(x)<S_1(x)$ for all $x>m$, the continuity of
$S_1$ and  $S_2$ implies the existence of $\delta >0$ so that 
$$
S_1(c^*)>S_2(c^*)+\delta
\mbox{ for all } c^* \in [S_{2}^{-1}(p_A),S_{1}^{-1}(p_A)]\mbox{ and for all } t>0.
$$
Thus by \eqref{eq3}, 
$$
\frac1{x_1}\frac{dx_1}{dt} > \frac1{x_2}\frac{dx_2}{dt} +
\delta\mbox{ for all } t>0,
$$
so $\ln x_1-\ln x_2\ge \delta t+\alpha$,
which implies that  $\frac{x_1}{x_1+x_2}\to 1$ as $t\to\infty$, completing the proof
of (i).

The proof of (ii) is analogous. 
\end{proof}

The next example
 is a numerical application of Theorem \ref{propB1}  and Proposition \ref{newPropA} where the desirability values decrease exponentially from the common median. Note that, as was the case in Example~\ref{exam2} above, the solution, in this case of the underlying coupled system of ODE's, is not available in closed form, but must be approximated numerically.

\begin{exam}\label{exam3} 
Let the survival functions $S_1$  and $S_2$ for subpopulations 
$B_1$ and $B_2$ be Laplace distributions with 
$S_1(x)={e^{-x}}/2$ for $x\ge 0$ and $S_2(x)={e^{-2x}}/2$ for $x\ge 0$ (see Figure~\ref{newFig4}). (For ease of exposition, the common median desirability of both subpopulations here is taken to be 0; translation to a different median value is trivial.) By Proposition~\ref{newPropA}, $P_1 \succ P_2$, so subpopulation $B_1$ is  more variable than  $B_2$.

\begin{figure}[!ht] 
  \center\includegraphics[width=0.75\textwidth]{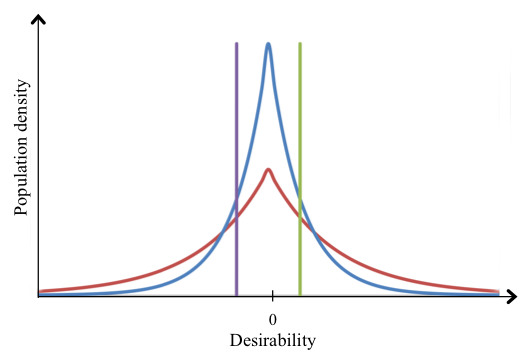}
  \caption{The red curve is the density of the desirability value of the more variable subpopulation $B_1$ in Example~\ref{exam3}, and the blue curve is the density of the less variable subpopulation $B_2$.  If $B_1$ and $B_2$ are of equal size, then the vertical green line is the threshold cutoff for the opposite sex $A$ so that exactly 25\%  of the composite $B$ population has desirability value above that point. The vertical purple line is the value so that exactly 75\% of the  $B$ population has desirability value above that point. Note that the desirability values of both drop off exponentially fast from the mean in both directions.}
 \label{newFig4}
\end{figure}

Suppose first that sex $A$ is selective and accepts only the most desirable quarter of individuals in sex $B$, i.e., $p_A=0.25$. Using \eqref{eq2} and \eqref{eq3}, and noting that $S_2(x)=2S^2_1(x)$ for $x\ge 0$ yields the following coupled system of ordinary differential equations: 

\begin{equation}\label{eq9}
\begin{split}
\frac{dx_1}{dt} &= x_1\left(\frac{\sqrt{x^2_1+2x_1x_2+2x^2_2} - x_1}
{4x_2}\right)\\
\frac{dx_2}{dt} 
&=\left(\frac{x_1+x_2}4\right)-x_1\left(\frac{\sqrt{x^2_1+2x_1x_2+2x^2_2}-x_1}
{4x_2}\right).
\end{split}
\end{equation}

No closed-form solution of  \eqref{eq9} is known, and Figure~\ref{fig5} illustrates a numerical solution with the initial condition $x_1(0)=x_2(0)=1$. Note that in this case where sex $A$ is selective, the more variable subpopulation $B_1$ eventually eclipses the less variable $B_2$.

\begin{figure}[!ht]  
  \center\includegraphics[width=1.0\textwidth]{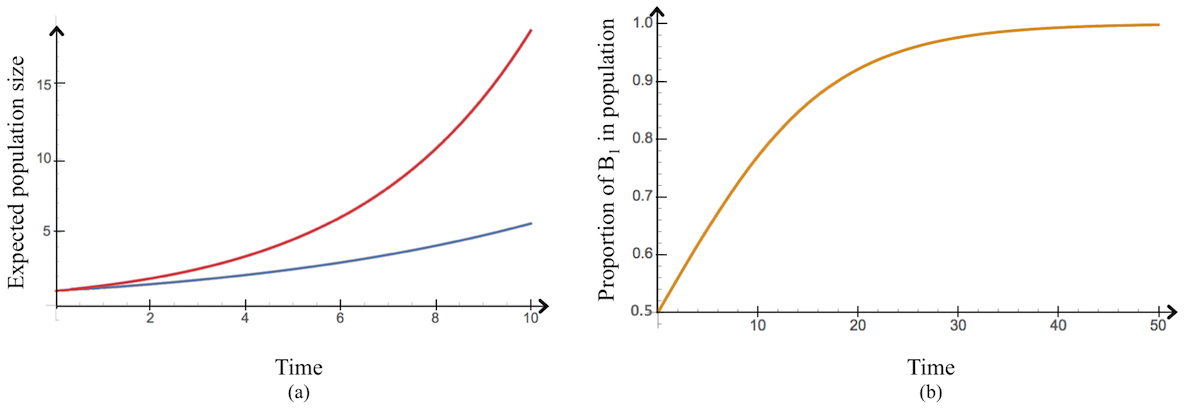}
\caption{Selective case---Population sizes and ratio. The graphs (a) of the more variable $x_1(t)$ in red and $x_2(t)$ in blue, and (b) the ratio $x_1(t)/(x_1(t)+x_2(t))$ satisfying (\ref{eq9}).
}
\label{fig5}
\end{figure}

Suppose next that sex $A$ is non-selective and accepts only individuals in in the most desirable three-quarters of sex $B$, i.e., $p_A=0.75$. Using \eqref{eq2} and \eqref{eq3} again, and noting that  $S_2(x) = 4S_1(x) - 2S^2_1(x) - 1$ for $x \leq 0$ yields the following system:

\begin{equation}\label{eq10}
\begin{split}
\frac{dx_1}{dt} &= x_1\left(\frac{x_1+4x_2-\sqrt{x^2_1+2x_1x_2+2x^2_2}}
{4x_2}\right)\\
\frac{dx_2}{dt}&=\left(\frac{3x_1+3x_2}4\right)-x_1\left(\frac{x_1+4x_2-
\sqrt{x^2_1+2x_1x_2+2x^2_2}}
{4x_2}\right).
\end{split}
\end{equation}

Figure~\ref{fig6} illustrates a numerical solution of the coupled system of ODE's \eqref{eq10} with the same initial condition $x_1(0)=x_2(0)=1$.  Note that in this situation where sex $A$ is non-selective, the less variable subpopulation $B_2$ eventually eclipses the more variable $B_1$.

\begin{figure}[!ht] 
  \center\includegraphics[width=1.0\textwidth]{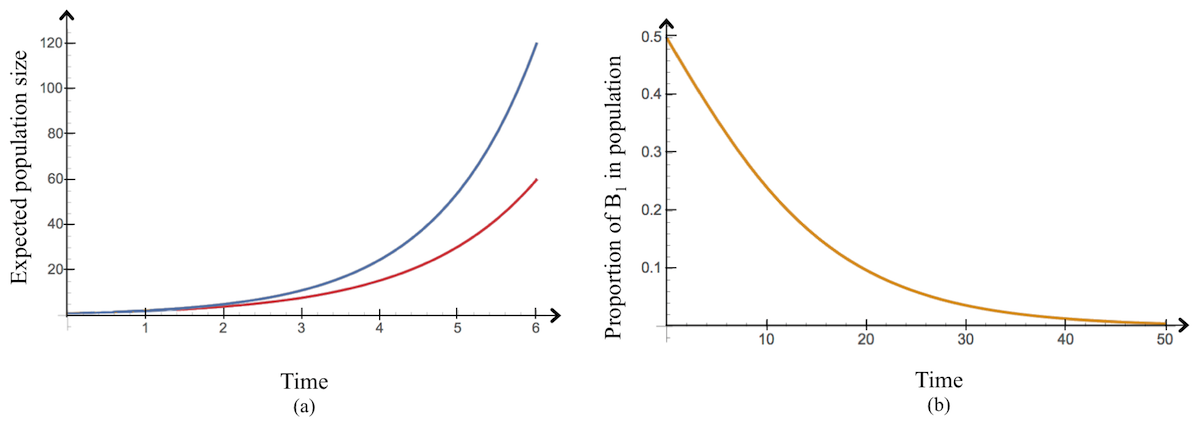}
\caption{Non-selective case---Population sizes and ratio.
The graphs (a) of the more variable subpopulation size $x_1(t)$ in red and less variable $x_2(t)$ in blue, and (b) the ratio $x_1(t)/(x_1(t)+x_2(t))$  satisfying (\ref{eq10}).
}
\label{fig6}
\end{figure}

\end{exam}

\noindent
Note also that the birth process model above also implicitly includes simple {\em birth-death} processes, via the simple observation that a population growing, for example, at a rate of eight per cent and dying at a rate of three percent, can be viewed as a pure birth process growing at a rate of five per cent.

\section{Darwin's Question, Selectivity and Parenting}

As mentioned above, recent research has generally confirmed Darwin's observation of greater male variability for many species and traits. 
For example, in the realm of human studies, 
\begin{quote}
The principal finding is that human intrasex variability is significantly higher in males, and consequently constitutes a fundamental sex difference...The data presented here show that human greater male intrasex variability is not limited to intelligence test scores, and suggest that generally greater intrasex variability among males is a fundamental aspect of the differences between sexes. Birth weight, blood parameters, juvenile physical performance, and university grades are parameters which reflect many aspects of human biology. In particular, the differences in variations in birth weight strongly suggest that social factors cannot account for all of the sex differences in variability \cite[pp.~198, 204--205]{vob5}.
\end{quote} 
The selectivity-variability principle introduced above is neutral with respect to the two sexes, and by itself does not explain any differences in variability between them - either that there should be a difference in variability between the sexes, or which sex might be expected to be more variable. But together with two other basic biological tenets, the selectivity-variability principle can perhaps help provide a theory for Darwin's observation and the empirical evidence of greater male variability reported in many subsequent studies.

One of these two additional biological tenets is parenting-selectivity, which posits that a ``basic cross-species pattern is that the sex with the slower potential rate of reproduction invests more in parenting, [and] is selective in mate choices [and the] sex with the faster potential rate of reproduction invests less in parenting, [and] is less selective in mate choices'' \cite[p.~273]{vob82}.
For example, ``When females invest more in offspring than males, parental investment theory says that selection operates so that females discriminate among males for mates...and males are indiscriminate''  \cite[p.~2037]{novob5}; (see also \cite{novob6} and \cite{novob9}). Although the genetic mechanisms of pre- and postcopulatory sexual selection are still far from being fully understood, molecular genetic and genomic tools now enable their detailed experimental testing \cite[p.~300]{novob7}.

The second additional biological tenet is gender-parenting, which says that ``typical species [have] less parental investment by males than females''  \cite[p.~2037]{novob5} which occurs, for example, in more than ninety-five percent of mammalian species  \cite[p.~175]{vob3}. 
Combining these two biological maxims with the selectivity-variability principle suggests an answer to Darwin's question. By the gender-parenting tenet, females in typical species invest more in parenting than males, so by the parenting-selectivity tenet females will typically be relatively selective and males relatively non-selective.  Then the selectivity-variability principle implies that females in such species will tend toward less variability and the males toward greater variability. 

If both sexes in a certain species began with comparable mid-range variability, for example, and if either its females were generally selective 
$(p_F<0.5)$  or its males were generally non-selective $(p_M>0.5)$, or both, this would  have led to the relatively greater male variability observed by Darwin. In the constraints of this cross-species model, therefore, this would offer {\em two independent explanations} for the appearance of greater male variability in many species. Unlike other species, of course, in humans, cultural factors may also play a role in the perceived differences in variability between the sexes.

\section{Conclusions}

The goal here has been neither to challenge nor to confirm Darwin's and other researchers' observations of greater male variability  for any given species or any given trait, but rather to propose an elementary mathematical theory based on biological/evolutionary mechanisms that might serve as a starting point to help explain how one sex of a species might tend to evolve with greater variability than the other sex. As such, the contribution here is a general theory intended to open the discussion to further mathematical and statistical modeling and analysis.

\section{Further research directions}

There are many natural generalizations, modifications, and extensions of the basic selectivity-variability theory described above, including the following: 
\\

\noindent 
{\em Time dependence}. Desirability distributions and/or selectivity that vary with time, for example individual desirability decreasing with time and the proportion of acceptable mates increasing with time.
\\

\noindent
{\em Desirability}. Desirability that depends on several parameters (e.g., size and intelligence); or desirability that is a random variable (e.g., the perception of an individualÕs desirability  by the opposite sex is not perfect.)
\\

\noindent
{\em Selectivity criteria}. Acceptability of potential mates of the opposite sex that depends both on desirability  and proximity (e.g., nearby candidates may be more acceptable than distant candidates with higher desirability ); desirability  thresholds that vary within members of the same subpopulation; or acceptability criterion that is not a step function, but becomes continuously higher with higher desirability levels (similar to replicator equation models).
\\

\noindent
{\em Game-theoretic versions}. Subpopulations that compete and may choose their own selectivity cutoff thresholds - e.g., may choose to be less selective to increase their probability of having offspring.
\\

\noindent
{\em Non-identical means}. Competing subpopulations with unequal means, and the relative advantages between having a higher mean and lower variance. 
\\

\noindent
{\em Offspring}. Expected numbers of offspring that depend on desirability levels of both parents; offspring types that are randomly distributed with desirability and selectivity criteria depending on those of the parents, both models where each sex's selectivity and/or variability depend only on that of its parent of the same sex, and models where they depend on the those of both parents.
\\

\noindent
{\em Variability}.  Alternative notions of variability that are defined via standard deviation after outliers are removed, or are one-sided (e.g., only upper tail comparisons are important). 
\\

\noindent
{\em Cultural factors}. There are also many cultural aspects of research on the variability hypothesis suggested by empirical evidence of greater male variability in humans and by the selectivity-variability theory above. These include the effects of monogamy, education, religion, social status, etc. on differences in selectivity and variability between the sexes, and the rate at which the disparity in variability between the sexes is disappearing, as predicted by this theory and observed in empirical studies. 
\\

\noindent
{\em Univariate applications}. The selectivity-variability principle may also be applied to univariate decision-making. For example, the so-called ``Texas Top Ten Per Cent Law" guarantees Texas students who graduate in the top ten percent of their high school class automatic admission to all state-funded universities. The selectivity-variability principle implies that if a student has a choice of tutoring options which have similar average success rates, options with greater variability will be superior to those with lesser variability.
\\

\noindent
{\em Non-biological applications}. One colleague has suggested that a similar selectivity/variability principle may also apply to some chemical or quantum processes where two reagents interact, and one (or both) may have several different forms that vary in affinity for the other reagent.

 \section*{Acknowledgments}
The author is grateful to several referees and to Professors Ron Fox, Igor Rivin, Erika Rogers, Marjorie Senechal, Sergei Tabachnikov, and  Julius Van Der Werf for very helpful comments.

\clearpage

\begin{appendices}

\renewcommand{\thesection}{\Alph{section}.}
\section{Empirical Evidence}

Recall that the variability hypothesis pertains to general traits in animal species with male and female sexes, and is a general hypothesis that clearly does not apply to every trait in every sexually dimorphic species. This century has already produced much new empirical evidence on the variability hypothesis in different contexts, and the following are direct excerpts from a representative selection of such studies, most of which refer to humans.  These are grouped chronologically first by primary studies with findings that support the variability hypothesis, and then by primary studies that do not support greater male variability, and finally, those that are mixed. (Note:  Readers are welcome to suggest additional references.)

\renewcommand{\thesection}{\Alph{section}}

\subsection{Primary Research Studies Supporting the Greater Male Variability Hypothesis}

\begin{enumerate}


\item
``Boys were over-represented at the low and high extremes of cognitive ability"  \cite[p.~533]{vob20}.


\item
``dispersion dimorphism was significantly present, with greater male variability, in the case of skinfolds" \citep{vob217}.


\item
``We found greater variance, by Levene's test of homogeneity of variance, among boys at every age except age two despite the girls' mean advantage from ages two to seven" \cite[p.~39]{vob17}.

\item
``for all three tests there were substantial sex differences in the standard deviation of scores, with greater variance among boys. Boys were over represented relative to girls at both the top and the bottom extremes for all tests, with the exception of the top 10\% in verbal reasoning \ldots In relation to sex differences in variability, the current results support the general finding of greater male variability" \cite[pp.~463, 475]{vob37}. 


\item
``Males have only a marginal advantage in mean levels \ldots but substantially greater variance" \cite[p.~451]{vob21}.

\item
 ``Males are more variable on most measures of quantitative and visuospatial ability, which necessarily results in more males at both high- and low-ability extremes \ldots The results from multiple large-scale studies have confirmed greater variability among males than among females in many cognitive domains, including on measures of mathematics, science, and spatial abilities" \cite[p.~1, 22]{vob25}.


\item
``within species size variation is significantly larger in males" \cite[p.~630]{vob105}.

\item
``our analyses show greater male variability, although the discrepancy in variances is not large \ldots There is evidence of slightly greater male variability in scores, although the causes remain unexplained" \cite[p.~495]{vob28}.

\item
 ``[we] reviewed the history of the hypothesis that general intelligence is more biologically variable in males than in females and presented data \ldots which in many ways are the most complete that have ever been compiled, [that] substantially support the [greater male variability] hypothesis" \cite[p.~529]{vob29}.

\item
``Males show greater variability than females in science" \cite[p.~122]{vob100}.

\item
``Gender gaps vary but differences in variabiltiy are a ROBUST PHENOMENON \ldots International testing results show greater variance in boys' scores than in girls' \ldots These results imply that gender differences in the variance of test scores are an international phenomenon and that they emerge in different institutional settings" \cite[pp.~1331--1332]{vob32}.


\item
 ``Despite the modest differences at the center of the distribution, the greater variability of male scores resulted in large asymmetries at the tails, with males out-numbering females by a ratio of 7 to 1 in the top 1\% on tests of mathematics and spatial reasoning \ldots greater male variance is observed even prior to the onset of preschool" \cite[pp.~220--221]{vob19}.

\item
``Boys had greater variability in these IQ scores" \cite[p.~42]{vob22}.

\item
``The principal finding is that human intrasex variability is significantly higher in males, and consequently constitutes a fundamental sex difference \ldots The data presented here show that human greater male intrasex variability is not limited to intelligence test scores, and suggest that generally greater intrasex variability among males is a fundamental aspect of the differences between sexes. Birth weight, blood parameters, juvenile physical performance, and university grades are parameters which reflect many aspects of human biology. In particular, the differences in variations in birth weight strongly suggest that social factors cannot account for all of the sex differences in variability" \cite[pp.~198, 204--205]{vob5}.

\item
``With one exception...all variance ratios were greater than 1.0" \cite[p.~395]{vob31}.

\item
``We found sex differences in the variance of reading achievement in all three studies analysed. The biggest variance ratio 1.20 was in the PISA 2003 study and it was 1.08 in both PIRLS studies. In the PISA study, boys showed greater variance in reading comprehension than the girls in all countries, and in the PIRLS studies, the boys' variance was larger in most countries" \cite[p.~11]{vob101}.


\item
``By age 10 the boys have a higher mean, greater variance and are over-represented in the high tail. Sex differences in variance emerge early -- even before pre-school -- suggesting that they are not determined by educational influences" \cite[p.~26]{vob72}.

\item
``The hypothesis of `greater male variability' was supported in most domains" \cite[p.~475]{vob132}.

\item
 ``The overall variance ratio in Study 1 was 1.07. That is, males displayed a somewhat larger variance \ldots In Study 2, the average variance ratio was 1.09" \citep{vob51}.


\item
``the variability analyses tended to support the Greater Male Variability Hypothesis" \cite[p.~807]{vob76}.

\item
``We find that male students exhibit greater residual as well as raw variability for this data set" \cite[p.~2947]{vob133}.


\item
``Seven international tests revealed that on average the variance for males was 12\% larger than that for females" \cite[p.~132]{vob77}.


\item
``the variances in the personality descriptions by informants were higher for male than for female targets" \cite[p.~142]{vob18}.

\item
``In sum, the data largely supported the greater male variability hypothesis for mathematics achievement and general student achievement  \ldots In most countries, boys demonstrated larger variability than girls in (manifest) performance scores" \cite[pp.~390--391]{vob126}.

\item
``Overall, the results were consistent with previous research, showing small mean differences in the three domains, but considerably greater variability for males" \cite[p.~263]{vob30}.


\item
``our research team [New Paltz Evolutionary Psychology Lab] has recently uncovered that in humans, sex differences in variability across many behavioral and physical domains ends up being the rule \ldots for a large number of variables, males demonstrate[d] greater variability" \citep{vob1}.

\item
 ``For mathematics performance, across three meta-analyses and a wide variety of samples, variance ratios consistently range between 1.05 and 1.20 [i.e., males consistently have between five and twenty percent higher variance than females] \ldots Similarly, for verbal performance, variance ratios range between 1.03 and 1.16" \cite[p.~390]{vob84}. 

\item
``This study weighs in on a number of hypotheses related to the nature of sex differences in broad and narrow/specific cognitive abilities \ldots This research also shows, quite compellingly, that the variability hypothesis is plausible and impacts both manifest and latent analyses of general ability'' \cite[p.~50]{vob182}.

\item
``Analysis of 5,772,811 examinations from 1990 to 2009 revealed a switch in average performance, from men performing equally well or better than women to women performing better than men. However, greater male than female variability in performance did not change" \cite[p.~315]{vob189}.

\item
 ``The larger score variance of boys in general knowledge and most domains observed in the current study is consistent with prior evidence regarding general intelligence" \cite[p.~10]{vob44}.

\item
``Levene's test indicated that high school male students had greater variability in general knowledge. This is in accordance with all previous Croatian studies of gender differences in general knowledge" \cite[p.~132]{vob185}.


\item
``the gender difference for Russian high school students in mean mathematical test scores is negligible (d=0.05). However, boys show a greater variance of test scores than girls (VR=1.12). Both findings are consistent with the results previously reported in the studies based on US and international data \ldots greater male variability in mathematical test scores that we find with the Russian data is a more robust phenomenon reported in many other studies" \cite[p.~14]{vob180}.

\item
``the results of the study supported the greater male variability hypothesis in urban and rural samples \ldots the results of the present study found that the greater male variability hypothesis in creativity was consistent across different samples" \cite[pp.~85, 88]{vob13}.

\item
``Consistent with previous research, the variability of boys' performance in science was larger than that of girls' \ldots Variance ratios across all grades exceeded Feingold's (1994) criterion for greater male variability and were comparable to that found for mathematics. These variance ratios were also stable across the time period examined, with no association with year of assessment or interaction with grade" \cite[p.~651]{vob36}.


\item
``Twelve databases from IEA [International Association for the Evaluation of Educational Achievement] and PISA [Program for International Student Assessment] were used
to analyze gender differences within an international perspective from 1995 to 2015 \ldots The `greater male variability hypothesis' is confirmed" \cite[p.~1]{vob10}.

\item
``The distribution of general intelligence is broader in men than in women while neither is smarter on average" \cite[p.~iii]{vob54}.

\item
``Although mean scores of men and women did not differ \ldots a significant difference in variability of scores was observed \ldots The effect size was large (VR = 1.82) and statistically significant" \cite[p.~468]{vob124}.

\item
``One unexpected observation is that the excess variability of males relative to females tends to be greater at higher levels of socio-economic development \ldots Also, greater male variability is a near-universal finding. The similarities of sex differences in this world-wide sample of countries are difficult to explain with cultural or economic causes alone. They are more compatible with biologically based differences that show little variation between human populations" \cite[p.~247]{vob183}.


\item
``Our analysis suggests that plastic traits, such as size, and especially those under strong sexual selection, may indeed show more variation in males, when males are the bigger sex" \cite[p.~9]{vob61}.
 
\item
``There was generally greater male variance across structural measures [in the human brain]" \cite[p.~2]{vob49}.

\item
``In absolute and percentage terms, males are more represented in the upper and lower score ranges \ldots The percentage of female students in the top score ranges has decreased on the new SAT\ldots proportionally 45\% more males are in the 1400--1600 [SAT] score range than females" \cite[p.~7]{novob31}.

\item
``On average, male variability is greater than female variability on a variety of measures of cognitive ability, personality traits, and interests \ldots This finding is consistent across decades \ldots There is good evidence that men are more variable on a variety of traits, meaning that they are over-represented at both tails of the distribution (i.e., more men at the very bottom, and at the very top), even though there is no gender difference on average" \citep{vob9}.

\item
``We observed significantly greater male than female variance for several key brain structures, including cerebral white matter and cortex, hippocampus, pallidum, putamen, and cerebellar cortex volumes" \cite[p.~1]{vob47}.

 
\item
``greater male variability [was] found based on mean and variability analyses, respectively" \cite[p.~1]{vob135}.

\item
``this study concluded that males' variance was much larger than females" \cite[p.~579]{vob128}.

\item
``Here, we use recent meta-analytic advances to compare gender differences in academic grades from over 1.6 million students. In line with previous studies we find strong evidence for lower variation among girls than boys" \cite[p.~1]{vob85}.


\item
``boys were found to have more variability in mathematics achievement than girls" \cite[p.~27]{vob153}.

\item
``the performance of boys was more variable than that of girls in most nations, consistent with the greater male variability hypothesis \ldots the performance of males showed greater variability than for females, which leads to a higher proportion of high-achieving male students in STEM" \cite[p.~25]{vob129}.

\item
 ``we replicate and expand Baye and MonseurÕs work, and \ldots broadly confirm that variability is greater for males internationally" \citep{vob123}.


\item
``For both tasks, male and female mean performance was similar across four years of testing; however, males did exhibit a wider range of variation in performance on the reversal spatial task compared with females" \citep{vob174}.

\item
``These results provide evidence that greater male neuroanatomical variability extends beyond humans, and suggest both evolutionary and developmental explanations for this phenomenon" \cite[p.~1]{vob173}.

\item
``In n = 3069 participants, from 8 to 95 years of age, we found widespread greater variability in males compared with females in cortical surface area and global and subcortical volumes for discrete brain regions \ldots In conclusion, we demonstrated that males are more variable compared with females in individual brain regions with regard to surface area and volume measures" \cite[pp.~5420, 5428]{vob178}.

\item
 ``Our meta-analysis of social-dilemma studies based on 40 samples with 8,123 participants strongly supports our hypotheses and provides empirical evidence of greater male variability in cooperation" \cite[p.~9]{vob172}.

\item
 ``the largest-ever mega-analysis of sex differences in variability of brain structure, based on international data spanning nine decades of life \ldots The present study included a large lifespan sample and robustly confirmed previous findings of greater male variance in brain structure in humans. We found greater male variance in all brain measures, including subcortical volumes and regional cortical surface area and thickness, at both the upper and the lower end of the distributions" \cite[pp.~6, 23]{vob167}.


\item
``the results of variability analyses suggest a pattern of greater male variability: (1) men exhibited greater variance than women in the overall distribution of creative self-efficacy scores, with the male/female variance ratio (VR) = 1.64, and (2) men were overrepresented at both the high and low extremes of the score distribution, with the male/female ratio = 2.62 Ð 9.99" \citep{vob216}.

\item
``we provide modest support for the `greater male variability hypothesis' \ldots using measures that are highly unlikely to be subject to social or cultural influences" \citep[p.~26]{vob226}

\item
 ``Our observation is consistent with the greater male variability hypothesis" \citep[p.~5]{vob242}.
 
 \item
``as predicted by Wainer's rule, males present overall more variance in size and shape, albeit this is statistically significant only for total cranial size" \cite[p.~2789]{vob231}.

 \item
``In a meta-analysis of experimental economics studies with more than 50,000 individuals in 97 samples, we find converging evidence for greater male variability in time, risk, and social preferences" \citep{vob219}.

\item
``Our findings thus generally support greater male variability, a phenomenon which is known to extend beyond brain measures \cite[p.~4635]{vob230}.


\item
``The results confirmed the existence of greater male variability (GMV) in DT" \citep{vob225}.

\item
 ``This represents evidence of GMV for brain size in a non-human primate species" \citep[p.~1]{vob237}.
 
 \item
``In the two-person beauty contest \ldots we observe greater male variability in strategic sophistication among males \ldots. As a result, there are fewer women among the bottom and top performers in the game" \cite[p.~288]{vob209}.
 
 \item
``The finding that boys demonstrated greater heterogeneity in social skills at nearly every grade level (except in first and second grade) is consistent with past research. This finding is not altogether different from findings supporting greater male variability in other psychological and educational variables such as intelligence, achievement, or personality" \citep[pp.~497-98]{vob254}.
 
\item
``We analyzed a large database on energy expenditure in adult humans....[and] found that even when statistically comparing males and females of the same age, height, and body composition, there is much more variation in total, activity, and basal energy expenditure among males \ldots Our results indicate considerable GMV in human energy expenditure in terms of TEE, BEE, and AEE...We also found GMV in key measures of body condition associated with energy expenditure \citep[p.~1]{vob223}.

\item
``We found no evidence for greater female variability on any measure. Overall, males had greater variability than unstaged females" \citep[p.~1]{vob236}.

\item
``Notably, in the dataset the ultimately highest and lowest-performing pupils were boys. Also, in both extremes (scores $<$ 200 and $>$800), the number of boys is twice that of girls. This fits with the greater male variability hypothesis" \cite[p.~70]{vob227}.

\item
``A consistent pattern of GMV in hearing thresholds was seen across frequencies in both datasets. In addition, both across-ear and within-ear correlations between thresholds were consistently greater in males than females" \citep[p.~1]{vob238}.

\item
``There were no factors with higher female than male standard deviation, therefore also the total score had greater male variability ($\Delta$ SD = 1.90; F = 6.031; p = .014)" \cite[p.~276]{vob221}.


\item
``we found limited support for hypotheses regarding greater male-than-female variability" \citep[p.~1]{vob247}.

\item
``total and perhaps also activity EE, associated with many morphological and physiological traits combined, do exhibit GMV most prominently during the reproductive life stages" \citep{vob258}.

\item
 ``Results confirmed greater male variability in thing orientation, but positive skewness for women" \citep[p.~1]{vob239}.


\item

 ``The variability was higher in males/men than in females/women in the total level of autistic traits, social behaviors/social skills, communication/mindreading, attention to details/patterns, and attention switching/tolerance of change in a clinical and a non-clinical group of women" \citep[p.~1]{vob240}.

\item
``Based on the present and reviewed studies, the conclusion must be that greater male variability contributes substantially to males' greater productivity, in terms of a larger proportion of male researchers who are very productive" \citep[p.~134]{vob241}.

\item
``In waxbills \ldots wider variation in male lateralization agrees with the `greater male variability hypothesis'" \citep[p.~50]{vob243}.

\end{enumerate}


\subsection{Primary Research Studies Not Supporting the Greater Male Variability Hypothesis}

\begin{enumerate}


\item
  ``the finding in this meta-analysis [is] that there is no sex difference in variance on the Advanced Progressive Matrices and that females show greater variance on the Standard Progressive Matrices"
\cite[p.~520]{vob70}.


\item
``the common assumption that males have greater variance in mathematics achievement is not universally true'' \cite[p.~S152]{vob80}.


\item
``data from several studies indicate that greater male variability 
with respect to mathematics is not ubiquitous...[and] is largely an 
artifact of changeable sociocultural factors, not immutable, 
innate biological differences between the sexes \ldots Our finding \ldots [is] inconsistent with the 
Greater Male Variability Hypothesis" \cite[pp.~8801, 8806]{vob41}.


\item
``Boys were not found to be more variable than girls" \cite[p.~326]{vob2}.
 
\item
``Therefore, we conclude that both variance and 
VR [variance ratio - ratio of male to female variance] 
in mathematics performance vary greatly among countries \ldots These 
findings are inconsistent with the greater male variability 
hypothesis" \cite[p.~14]{vob68}.


\item
``When differences in variance are considered, the data again shows little support for any sex differences, in either general or specific ability scores" \cite[p.~59]{vob156}.


\item
``Levene's test of homogeneity of variances (Table 2) showed that the variance was homogeneous between 
boys and girls" \cite[p.~5]{vob139}.

\item
``the `greater male variability hypothesis' is not supported by our data" \cite[p.~13]{vob151}.

\item
``in the surrounding areas of Bangkok \ldots females consistently outperformed their male counterparts in both English language and science related subjects, and also outnumbered their male peers in the top-100 achievers in both domains" \cite[p.~17]{vob176}.


\item
``we find no evidence for widespread sex differences in variability in non-human animal personality'' \citep{vob215}.

\item
``We find little or no support for the greater male variability hypothesis" \citep{vob228}.

\item
``In two studies, we found little evidence for gender differences in creative variability for verbal and  gural task performance in adults and adolescents" \citep{vob218}.


\item
``Regarding variability per-se, it is found to be very close between sexes and more often than not, greater for women. The prudent conclusion is thus one of gender similarity in the variability of human characteristics" \citep[p.~1]{vob244}.

\item
``The greater male variability hypothesis in performance was not supported for any of the performance measures" \citep[p.~1]{vob246}.



\item
``we found distinct greater female variability in human colour vision" \citep[p.~2]{vob255}.

\item
``Taken together, gender differences in the mean and variability of creative ability scores are minimal and inconsistent across different contexts, suggesting that the GMVH may not provide much explanatory power for the gender gap in creative achievement" \citep{vob257}.

\end{enumerate}


\subsection{Primary Research Studies with Mixed Results}

\begin{enumerate}


\item
``Variation was significantly greater among men than women in 5 of the 6 former data sets and was similar for men and women in the latter 2 data sets, broadly supporting the predictions. A further analysis extends the theory to intellectual abilities" \cite[p.~219]{vob65}. 


\item
``Our data show non-existent or trivial gender difference in mean scores. However, the tails of the distributions show differences between the males and females, with greater variability among males in the upper half of the distribution and greater variability among females in the lower half of the distribution" \cite[p.~292]{vob148}


\item
``all of the boy/girl VRs of the TCT-DP scores, with the exception of the Humor and Affectivity subscale, were above 1.0, which suggests that boys have greater variance than girls in creativity test performance \ldots In the lower tails, the boy/girl ratio for the regions where $z \leq -1$ (i.e., 1 standard deviation or more below the mean) was 0.92, indicating that more girls fell into this region, which is not in line with the prediction of the greater male variability hypothesis" \cite[pp.~884]{vob26}. 


\item
``The results of F tests of equality of variance revealed a statistically significant difference regarding the greater female variance in the children group and the greater male variance in the adolescent and the emerging adult groups" \cite[pp.~93]{vob140}.

\item
``this study showed that the Greater Male Variability Hypothesis had supportive evidence from responses to figural stimuli but not responses to verbal stimuli of the WKCT" \cite[p.~87]{vob127}.

\item
``I indeed found evidence of greater male variability, but only in Realistic and Enterprising interests and the higher-order PeopleÐThings dimension. Moreover, I established greater female variability in Artistic and Conventional interests " \cite[p.~575]{vob35}.


\item
``Study 1 demonstrated not only that greater male variability replicates among young children, but also that the general effect size of this ratio should be considered large (VR $>$ 2) \ldots Study 2 has confirmed these expectations: Not only were males characterized by a higher variability of original thinking and unconventionality, but also their mean originality scores were higher than those of females. By contrast, females were characterized by a higher mean level and variability of adaptiveness" \cite[p.~164]{vob99}.


\item
``Data for a sample of 18-year-old Applied Science and Social Science students in Libya are reported for intelligence (IQ), emotional intelligence (EQ) and educational achievement \ldots females had greater variability of intelligence than males \ldots In our study there were no clear sex differences in the variability of academic performance (Table 3), and only emotional intelligence showed greater male variability (Table 1)" \cite[pp.~448, 453]{vob179}.           


\item
``our results and conclusions provide strong evidence for the variability hypothesis in humans \ldots males appear to be more variable in the great majority of physical, cognitive, behavioral traits that have been investigated \ldots One notable exception to this general trend is that women appear to show higher levels of variability compared to men in some emotional traits, such as emotionality, anger, discomfort, fear, negative affectivity, and soothability" \cite[pp.~44, 53]{vob134}.

\item
``In most countries, male students indeed have greater variance in math performance (VR . 1.0), but the size of the gender difference in variances is small and varies by country, ranging from VR = 1.009 in Argentina to VR = 1.092 in Colombia. Moreover, in Brazil and Peru, female students actually have greater variances than male students (VR \ 1.0), which shows that the greater male variability hypothesis does not hold across the board and offers evidence of cross-country differences" \cite[p.~16]{vob152}.

\item
``it would appear there are gender differences in reading favoring girls across all levels of ability distribution, with these being small (and more similar) in the middle of the distribution but much larger (and impactful) at the tails. Furthermore, these gender differences are found in younger students, as well as older ones. We did not find strong support for the greater male variability hypothesis because the larger number of low scoring boys was offset by the higher number of high scoring girls \ldots At the lower-end of the ability distribution, boys were greatly overrepresented by a factor of 2 or more which grew slightly larger for older students. Just as with reading, there was a reversal of gender ratios for students attaining an advanced writing proficiency, with girls greatly overrepresented by a factor of 2 or more'' \cite[p.~14]{vob170}


\item
``Greater male variability was observed in Botswana, Japan, Lithuania, Singapore and United States. On the contrary, girls were more variable in their science success in Bahrain, Kuwait, Norway, Oman and Turkey" \cite[p.~194]{vob168}.

\item
``Our study  also highlighted that  boys  were likely to outperform in the higher score ranges, as well as lower score ranges.  Finally, we also found that the boys had significantly greater variabilities in science achievement and interest (variance ratio [VR] $>$ 1.1) while the girls had slight greater variability in creativity" \cite[p.~195]{vob164}.

\item
``While greater male variability was found in morphological traits, females were much more variable in immunological traits" \cite[p.~1]{vob177}. 


\item
``there is no variance difference for serial reversal learning, but greater female variability in spatial memory. These results differ from those found in food hoarders, which found greater male variability in reversal learning, but not spatial memory. This also is the first time greater female variance has been found for a cognitive ability" \citep[pp.~ii-iii]{vob250}.


\item
 ``we compared variance between the sexes for 50 morphological and physiological traits, analysing data from the NHANES database. Nearly half the traits did not exhibit sexual dimorphism in variation, while 18 exhibited greater female variation (GFV), indicating GFV does not dominate human characteristics. Only eight traits exhibited greater male variation (GMV), indicating GMV also does not dominate" \citep[p.~1]{vob245}.


\item
``Although males were more variable in several domains (e.g., gV, gC), females were more variable for processing speed (gS), and there were no VR differences for still other domains (e.g., short-term verbal memory)" \citep[p.~1]{vob251}.

\item
 ``Our results also partly supported The Greater Male Variability Hypothesis \ldots There were more than twice as many males as females in the lowest 10\% score percentile of AUT Fluency. However, contrary to some previous findings (He et al., 2015), we did not find an over-representation of males in the 0\%--25\% upper regions of the distribution. Moreover, the overall variability on this task was slightly higher for females" \citep[p.~58]{vob256}.

\end{enumerate}

\end{appendices}


\end{document}